\documentclass[11pt]{article}

\usepackage{amsmath}
\usepackage{amsthm}
\usepackage{amssymb}
\usepackage{amsfonts}
\usepackage{graphicx}

\usepackage{hyperref}
\usepackage{float}
\usepackage{color}
\def\f12{\frac 1 2}

\def\hh{\mathcal{H}^{+}}
\def\r{\mathcal{R}}

\def\f12{\frac 1 2}

\newcommand{\lapp}{\mbox{$\triangle \mkern-13mu /$\,}}

\newtheorem{mytheo}{Theorem}
\newtheorem{myp}{Proposition}

\begin{document}
\title{A note on instabilities of  extremal black holes\\ under scalar perturbations from afar}

\author{Stefanos Aretakis\thanks{Princeton University, Department of Mathematics, Fine Hall, Washington Road, Princeton, NJ 08544, USA.}\thanks{  Institute for Advanced Study, Einstein Drive, Princeton, NJ 08540, USA.}}

\date{December 2, 2012}
\maketitle

\begin{abstract}
In previous work of the author it was shown that instabilities of solutions to the wave equation develop asymptotically along the event horizon of extremal Kerr provided a certain expression $H_{0}$ of the initial data is non-trivial on the horizon.  In this note we  remove this restriction by showing that instabilities  develop even from initial data supported arbitrarily far away from the horizon (for which, in particular, $H_{0}=0$). The latter instabilities concern one order higher derivatives compared to the case where $H_{0}\neq 0$. The result also applies to extremal Reissner--Nordstrom.   This note was motivated by numerical analysis of Lucietti, Murata, Reall and Tanahashi. 

\end{abstract}

\section{Introduction}
\label{sec:Introduction}

The wave equation on black hole backgrounds has attracted significant interest of the mathematical community  during the past decade in view of its intimate relation with the stability of the spacetimes themselves in the context of the Cauchy problem for the Einstein equations. Specifically, one is interested in the decay properties of solutions $\psi$ to the wave equation
\begin{equation}
\Box_{g}\psi=0
\label{we}
\end{equation}
in the domain of outer communications up to and including the event horizon $\hh$.

 \begin{figure}[H]
	\centering
		\includegraphics[scale=0.127]{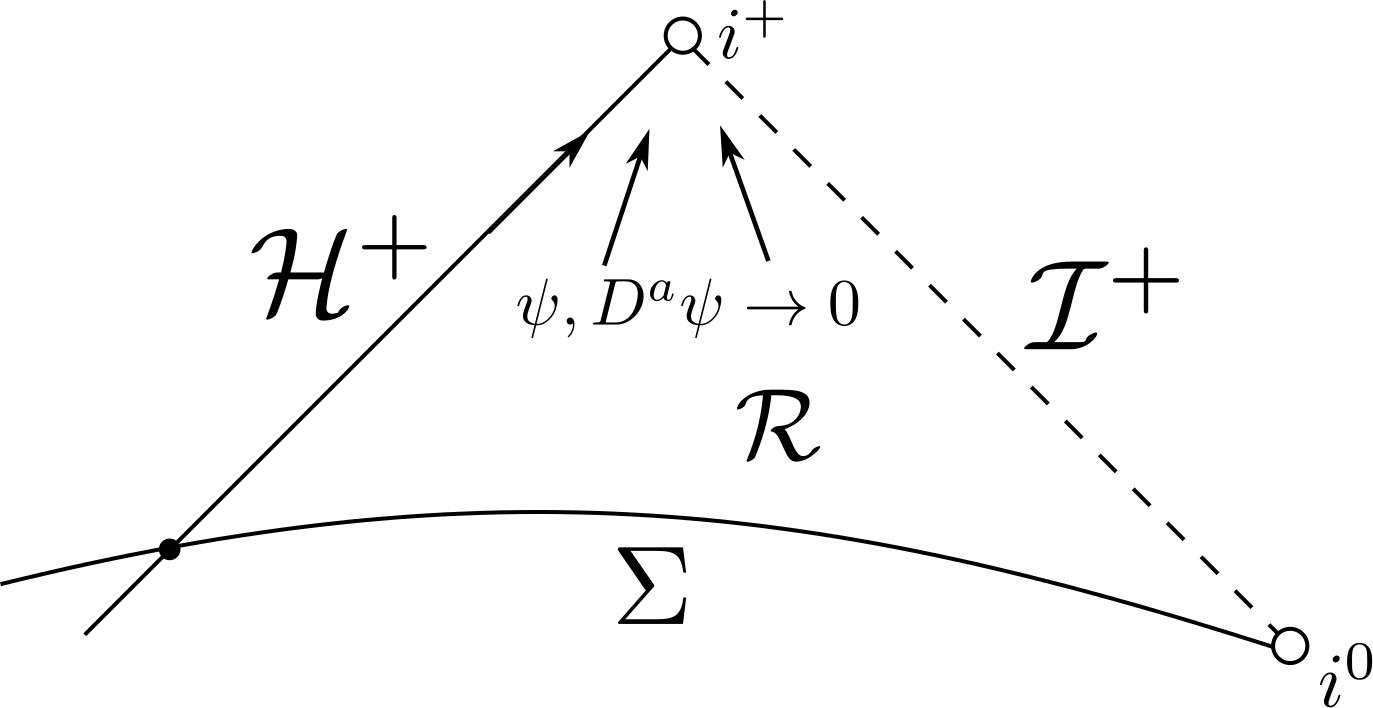}
	\label{fig:ekerr2cd1}
\end{figure}
The study of the wave equation on black holes goes back to Regge and Wheeler \cite{regge} in 1957. The first complete boundedness result for solutions $\psi$ to \eqref{we}, on Schwarzschild backgrounds, without any restriction of the initial data on $\hh$ was due to Kay and Wald \cite{wa1} (1987) and the first complete quantitative decay result  for such solutions appeared in 2005 \cite{redshift} by Dafermos and Rodnianski. The latter paper initiated the use of multiplier vector fields which capture the so-called \textit{redshift effect} along  the horizon $\hh$. 
 Since then impressive progress has been done for more general spacetimes with contributions by many authors; for a review, with many references, see \cite{lecturesMD}.  
  This work has in some sense culminated in \cite{tria} where uniform boundedness and decay estimates for $\psi$ and all its derivatives are derived (up to and including the event horizon) for general subextremal Kerr backgrounds with parameters $|a|<M$.

%
%

The wave equation on extremal black holes (for which the redshift effect degenerates at the horizon $\hh$) was much less studied until recently; see however \cite{other2,   blu1,jacobson, marolf}. The general study of the wave equation on such backgrounds was initiated by the author in a series of papers \cite{SA10, aretakis1, aretakis2, aretakis3,aretakis4}, where it was shown that solutions of \eqref{we} exhibit both stability and instability properties. Specifically, it was shown that general solutions on extremal Reissner--Nordstr\"{o}m and axisymmetric solutions on extremal Kerr  decay pointwise towards the future. On the other hand, it was also shown that, if $Y$ is a translation-invariant\footnote{Translation-invariant here means that $Y$ commutes with the stationary Killing field $T$, i.e.~$[Y,T]=0$.} transversal to $\hh$ derivative, then  \underline{for initial data for which a certain expression  is non-zero on $\hh$} then 
 \begin{equation}
\begin{split}
&\text{Non-decay: }\ \ \ \ \sup_{S_{v}}|Y\psi|\nrightarrow 0\\
&\text{Blow-up: } \ \ \ \ \sup_{S_{v}}|Y^{k}\psi|\rightarrow +\infty,\, k\geq 2, 
\end{split}
\label{pr}
\end{equation}
along $\hh$ as the advanced time $v\rightarrow +\infty$. Here $S_{v}$ denotes the (spherical) section of the horizon $\hh$ at advanced time $v$. As far as the instability results are concerned,  the axisymmetry assumption for solutions on extremal Kerr can be dropped by simply projecting on the zeroth azimuthal frequency.  

The origin of the instability results is the existence of conserved quantities along extremal horizons\footnote{These conservation laws have recently been extended to the Teukolsky equation on extremal Kerr by Lucietti and Reall \cite{hj2012}. The authors  have also provided extensions of the conservation laws for the wave equation to more general extremal horizons in all dimensions. Murata \cite{murata2012} has very recently provided further generalizations of the conservation laws.} $\hh$. Hence, if the conserved quantities are initially non-zero then they are everywhere non-zero on $\hh$. A question that was raised by Dain and Dotti \cite{dd2012} is whether the instabilities \eqref{pr} develop from initial data which are supported away from the horizon (and hence the conserved quantities are initially--and thus everywhere--zero on $\hh$). In the latter case, the authors of \cite{dd2012} derived a simple proof of the pointwise boundedness of $\psi$ on extremal Reissner--Nordstr\"{o}m by extending Wald's argument \cite{drimos}. Returning to the question of Dain and Dotti,   Lucietti, Murata, Reall and Tanahashi \cite{hm2012} performed numerical analysis which suggests that instabilities develop even from initial data supported away from the horizon of extremal Reissner--Nordstr\"{o}m.  Motivated by \cite{hm2012}, we rigorously show  the validity of this scenario which we in fact extend to extremal Kerr backgrounds. 

In this note we show the following

\begin{mytheo}
Consider the extremal Kerr black hole with parameters $a,M$ such that $|a|=M$. Let $\Sigma$ be a  spacelike hypersurface  which either crosses the event horizon and satisfies the assumptions of Section \ref{sec:TheInitialHypersurfaceSigma} or coincides with $t=0$, where $t$ denotes the Boyer--Lindquist time coordinate and $K=\Sigma\cap \left\{R_{1}\leq r\leq R_{2}\right\}$, where $M<R_{1}<R_{2}$. Let $Y$ denote a translation-invariant  vector field (that is $[Y,T]=0$, where $T$ is the stationary Killing vector field) transversal to $\hh$. Then, for generic smooth initial data supported in $K$ we obtain that the solution $\psi$ to the wave equation \eqref{we} satisfies:
\begin{equation*}
\begin{split}
&\text{Non-decay: }\ \ \ \ \sup_{S_{v}}|YY\psi|\nrightarrow 0\\
&\text{Blow-up: } \ \ \ \ \sup_{S_{v}}|Y^{k}\psi|\rightarrow +\infty,\, k\geq 3, 
\end{split}
\end{equation*}
along $\hh$ as the advanced time $v\rightarrow +\infty$.  Here $S_{v}$ denotes the (spherical) section of the horizon $\hh$ at advanced time $v$.
\label{theo1}
\end{mytheo}
 \begin{figure}[H]
	\centering
		\includegraphics[scale=0.16]{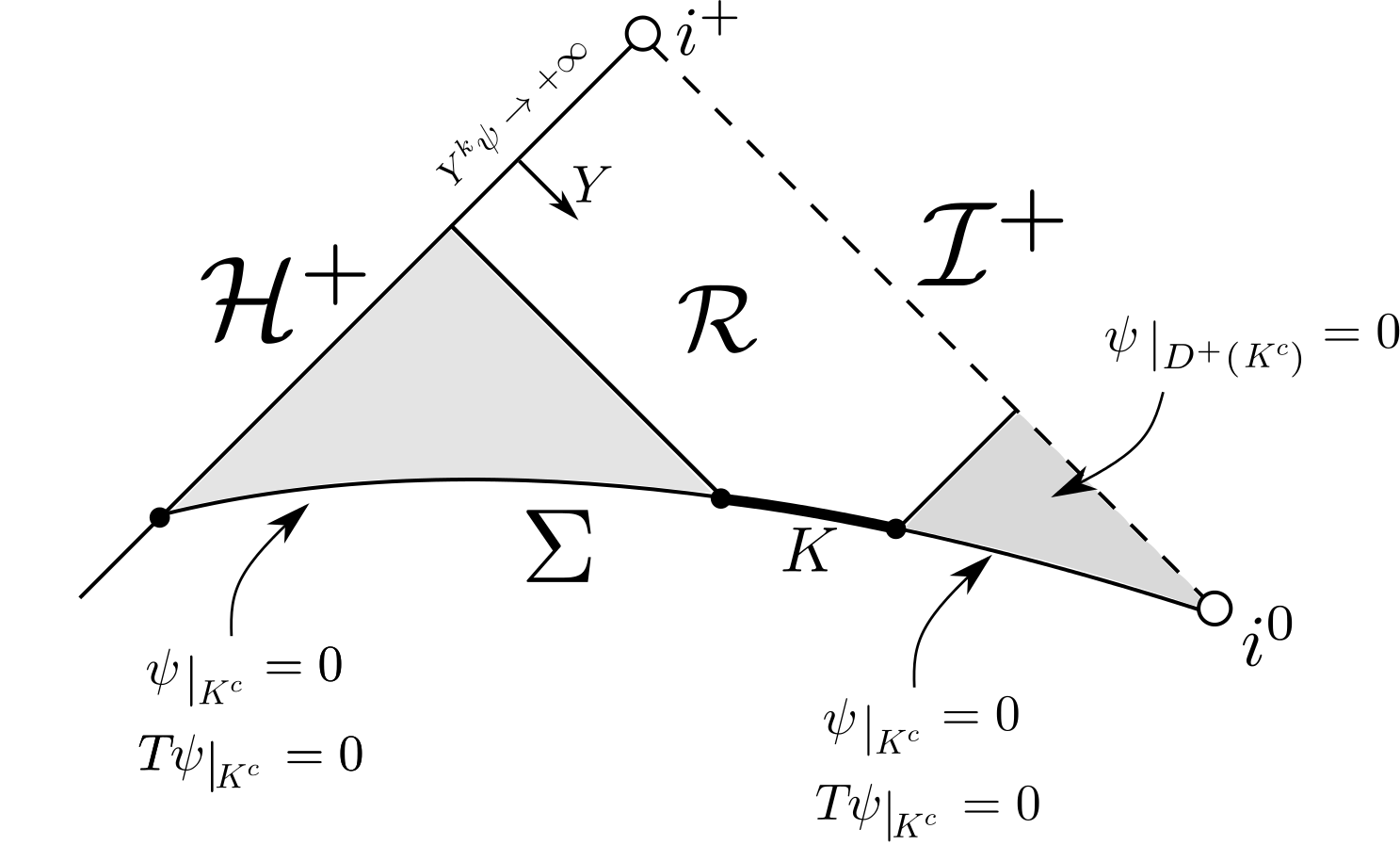}
	\label{fig:ekerr2cdpi2}
\end{figure}
In particular, $K$ may be chosen to be arbitrarily small and arbitrarily far away from the horizon. Note that the instabilities of Theorem \ref{theo1} concern one order higher derivatives compared to the previous results (see \eqref{pr}) for initial data whose support includes $\hh$. On the other  hand, note that our previous stability results and the vanishing of the conserved quantity $H_{0}[\psi]=0$ (see \cite{aretakis4} or Section \ref{sec:InstabilitiesForPsiAlongTheHorizon} below) imply that for  solutions $\psi$ as in Theorem \ref{theo1} we have
\[|\psi|\rightarrow 0,\   \ \ \ \ \int_{S_{v}}Y\psi\rightarrow 0\]
as $v\rightarrow+\infty$.  Note that the above, in particular, preclude the first non-decay inequality of \eqref{pr}.

We  remark that our argument can be readily adapted for the easier case of  extremal Reissner--Nordstr\"{o}m backgrounds. We finally mention an independent work of Bizon and Friedrich \cite{bizon2012} which provides a nice heuristic analysis of the problem based on the existence of conformal symmetries and heuristic analysis on $\mathcal{I}^{+}$.

\section{Geometric setup}
\label{sec:GeometricSetup}

\subsection{The Kerr metric}
\label{sec:TheKerrMetric}

The  Kerr metric with respect to the \textit{Boyer-Lindquist coordinates} $(t, r,\theta, \phi)$ is given by
\begin{equation*}
g=g_{tt}dt^{2}+g_{rr}dr^{2}+g_{\phi\phi}d\phi^{2}+g_{\theta\theta}d\theta^{2}+2g_{t\phi}dtd\phi,
\end{equation*}
where
\begin{equation*}
\begin{split}
&g_{tt}=-\frac{\Delta-a^{2}\sin^{2}\theta}{\rho^{2}}, \ \ \  g_{rr}=\frac{\rho^{2}}{\Delta},\  \ \   g_{t\phi}=-\frac{2Mar\sin^{2}\theta}{\rho^{2}},\\  
&\ \ \ \  \ \   g_{\phi\phi}=\frac{(r^{2}+a^{2})^{2}-a^{2}\Delta\sin^{2}\theta}{\rho^{2}}\sin^{2}\theta, \ \ \   g_{\theta\theta}=\rho^{2}
\end{split}
\end{equation*}
with
\begin{equation}
\Delta=r^{2}-2Mr+a^{2}, \ \ \ \ \ \rho^{2}=r^{2}+a^{2}\cos^{2}\theta.
\label{basic}
\end{equation}
Schwarzschild corresponds to the case $a=0$, subextremal Kerr to $|a|<M$ and extremal Kerr to $|a|=M$.

Note that the metric component $g_{rr}$ is singular precisely at the points where $\Delta=0$. To overcome this coordinate singularity  we introduce the following functions $r^{*}(r), \phi^{*}(\phi, r)$ and $v(t, r^{*})$ such that
\begin{equation*}
r^{*}=\int\frac{r^{2}+a^{2}}{\Delta}, \ \ \ \phi^{*}=\phi +\int\frac{a}{\Delta}, \ \ \ v=t+r^{*}
\end{equation*}
In the \textit{ingoing Eddington-Finkelstein  coordinates} $(v,r,\theta, \phi^{*})$  the metric takes the form
\begin{equation*}
g=g_{vv}dv^{2}+g_{rr}dr^{2}+g_{\phi^{*}\phi^{*}}(d\phi^{*})^{2}+g_{\theta\theta}d\theta^{2}+2g_{vr}dvdr+2g_{v\phi^{*}}dvd\phi^{*}+2g_{r\phi^{*}}drd\phi^{*},
\end{equation*}
where
\begin{equation}
\begin{split}
&g_{vv}=-\left(1-\frac{2Mr}{\rho^{2}}\right), \ \ \   g_{rr}=0,\  \ \  g_{\phi^{*}\phi^{*}}=g_{\phi\phi}, \  \ \  g_{\theta\theta}=\rho^{2}\\
\\& \ \ \ \  \ \  g_{vr}=1, \ \ \   g_{v\phi^{*}}=-\frac{2Mar\sin^{2}\theta}{\rho^{2}}, \  \ \  g_{r\phi^{*}}=-a\sin^{2}\theta.
\end{split}
\label{edi}
\end{equation}
For convenience we denote
\[T=\partial_{v},\ \ \ Y=\partial_{r}, \ \ \ \Phi=\partial_{\phi^{*}}.\]
For completeness, we include the computation for the inverse of the metric in $(v,r,\theta,\phi^{*})$ coordinates:
\begin{equation*}
\begin{split}
&g^{vv}=\frac{a^{2}\sin^{2}\theta}{\rho^{2}}, \ \ \   g^{rr}=\frac{\Delta}{\rho^{2}},\ \ \   g^{\phi^{*}\phi^{*}}=\frac{1}{\rho^{2}\sin^{2}\theta}, \  \ \  g^{\theta\theta}=\frac{1}{\rho^{2}}\\
\\& \ \ \ \ \ \ \  \ \   g^{vr}=\frac{r^{2}+a^{2}}{\rho^{2}} ,\  \ \   g^{v\phi^{*}}=\frac{a}{\rho^{2}}, \ \ \   g^{r\phi^{*}}=\frac{a}{\rho^{2}}.
\end{split}
\end{equation*}

Clearly, the metric expression \eqref{edi} does not break down at the points where $\Delta=0$, and in fact, it turns out  that this expression is regular even for $r<0$. On the other hand, the curvature would blow up at $\rho^{2}=0$, i.e.~the equatorial points of $r=0$.

Let $(\theta,\phi^{*})$ represent standard global\footnote{modulo the standard degeneration at $\theta=0,\pi$...} spherical coordinates on the sphere $\mathbb{S}^{2}$ and $S_{\text{eq}}$ denote the equator, i.e.~$S_{\text{eq}}=\mathbb{S}^{2}\cap\left\{\theta=\pi/2\right\}$. Let also $(v,r)$ be a global coordinate system on $\mathbb{R}\times\mathbb{R}$. We define the differential structure of the manifold $\mathcal{N}$ to be
\begin{equation*}
\mathcal{N}=\Bigg\{\big(v,r,\theta,\phi^{*}\big)\in\bigg\{\Big\{\mathbb{R}\times\mathbb{R}\times\mathbb{S}^{2}\Big\}\setminus\Big\{\mathbb{R}\times\left\{0\right\}\times S_{\text{eq}}\Big\}\bigg\}\Bigg\}.
\end{equation*}

From now on, we restrict our attention to extremal Kerr $|a|=M$, unless otherwise stated.  The event horizon $\hh$ is defined by
\begin{equation*}
\mathcal{H}^{+}=\mathcal{N}\cap\left\{r=M\right\}.
\end{equation*}
The \textit{black hole region} $\mathcal{N}_{\text{BH}}$ corresponds to
\begin{equation*}
\mathcal{N}_{\text{BH}}=\mathcal{N}\cap \left\{r<M\right\};
\end{equation*}
it is characterised by the fact that observers in the black hole region cannot send signals to observers located at points with $r>M$. The exterior region $\mathcal{D}$ given by
\begin{equation*}
\mathcal{D}=\mathcal{N}\cap \left\{r>M\right\}
\end{equation*}
is the so-called \textit{domain of outer communications}. This is precisely the region covered by  the Boyer-Lindquist coordinates.  Note that we shall be interested in studying the solutions to the wave equation in the region $\mathcal{D}\cup\hh$.

  The Penrose diagram of the extended region $\mathcal{N}$  is
 \begin{figure}[H]
	\centering
		\includegraphics[scale=0.135]{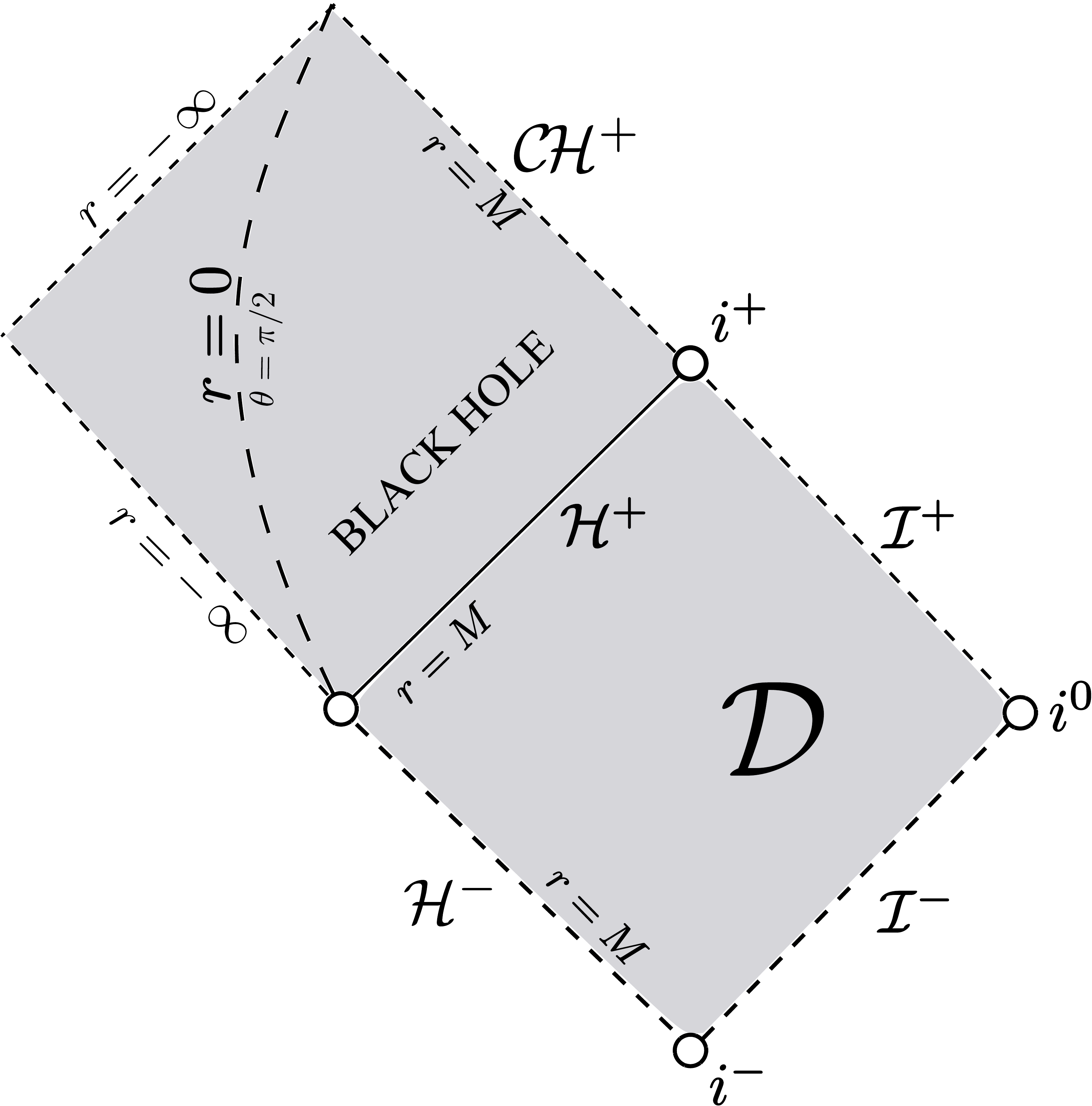}
	\label{fig:ekerr2}
\end{figure}

\subsection{The initial hypersurface $\Sigma$}
\label{sec:TheInitialHypersurfaceSigma}

Let $\Sigma$ be an axisymmetric hypersurface which crosses the horizon $\hh=\left\{r=M\right\}$, terminates at $i^{0}$, is everywhere transversal to the Killing field $T$ and also such that 
\begin{equation}
M^{2}\sin^{2}\theta-2(r^{2}+M^{2})h+\Delta h^{2}<0,
\label{condition}
\end{equation}
where $h(r,\theta)=-\frac{g(Y,n_{\Sigma})}{g(T,n_{\Sigma})}$ and $n_{\Sigma}$ denotes the future unit normal to $\Sigma$. Note that $g(T,n_{\Sigma})=g(T+\omega\Phi,n_{\Sigma})<0$ and also $g(Y,n_{\Sigma})=-g(-Y,n_{\Sigma})>0$. Note  that the above expression is negative for a neighborhood of $\hh$ and for large $r$ (for which $h\rightarrow 1$). We can hence choose $\Sigma$ such that the left hand side of \eqref{condition} is indeed negative everywhere on $\Sigma$ (note that this is  a very mild assumption on $\Sigma$). For example, if $\Sigma=\left\{t=0\right\}$ then $h=\frac{r^{2}+M^{2}}{(r-M)^{2}}$ from which it immediately follows that \eqref{condition} is satisfied.   Let now $M<R_{1}<R_{2}$.
Consider the following pieces
\[\Sigma_{1}=\Sigma\cap\left\{M\leq r\leq R_{1}\right\}, \ \ \Sigma_{2}=\Sigma\cap\left\{R_{1}\leq r\leq R_{2}\right\}, \ \ \Sigma_{3}=\Sigma\cap\left\{r\geq R_{2}\right\}\]
and the bit inside the black hole region
\[\Sigma_{0}=\Sigma\cap\left\{\frac{3}{4}M\leq r\leq M\right\}.\]
We also define
\[\r=D^{+}(\Sigma_{1}\cup\Sigma_{2}\cup\Sigma_{3})\cup\hh.\]
\begin{figure}[H]
	\centering
		\includegraphics[scale=0.16]{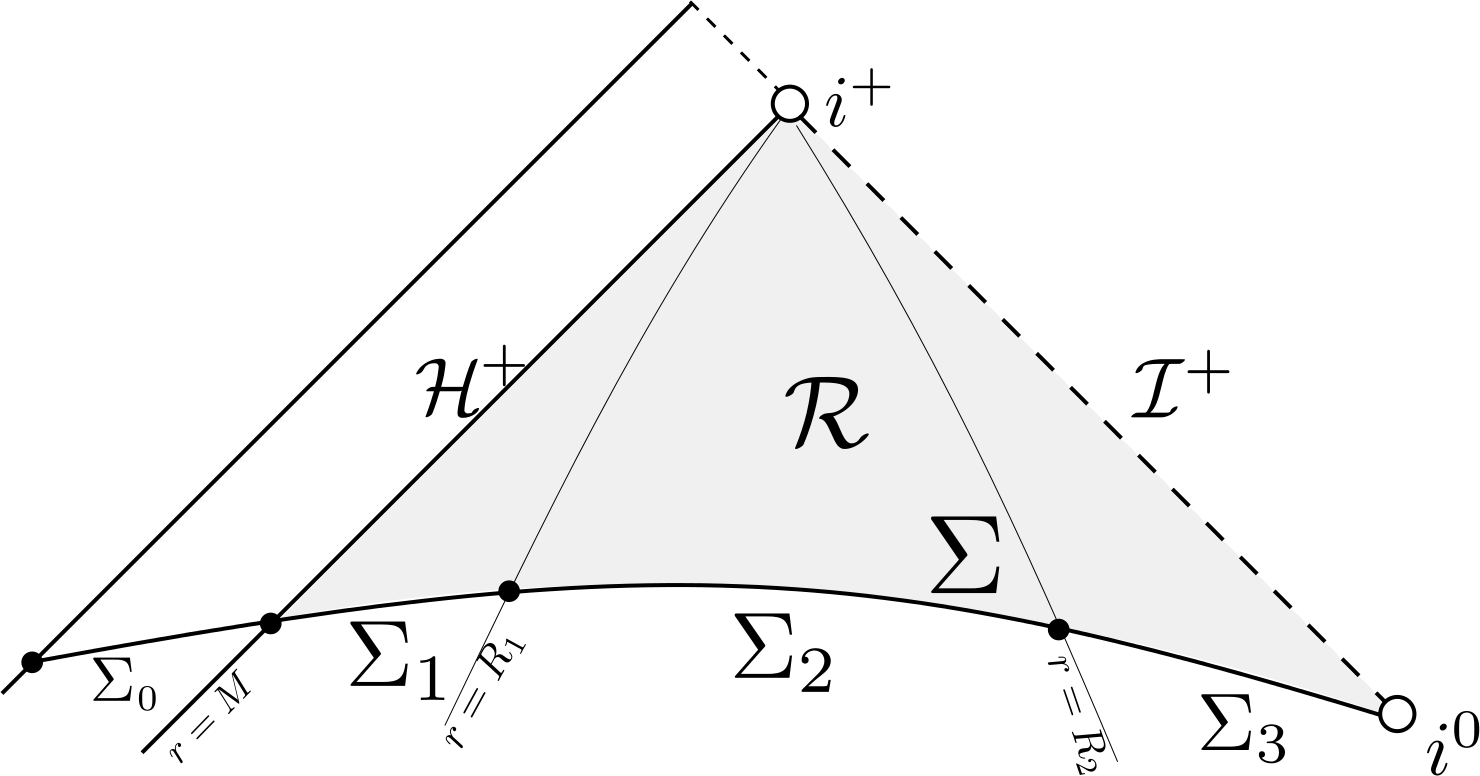}
	\label{fig:csind1}
\end{figure}

\subsection{The well-posedness of the wave equation}
\label{sec:TheWellPosednessOfTheWaveEquation}

Recall that $T$ is smooth in $D(\Sigma)$ and also everywhere transversal to $\Sigma$. Hence, given smooth initial data $(f_{1},f_{2})\in C^{\infty}(\Sigma)\times C^{\infty}(\Sigma)$, then there exists a unique solution to the Cauchy problem
\begin{equation*}
\begin{split}
&\Box_{g}\psi=0,\\
\left.\psi\right|_{\Sigma}=f_{1}&, \ \ \ \left.T\psi\right|_{\Sigma}=f_{2},
\end{split}
\end{equation*}
which is smooth in the interior of  $D^{+}(\Sigma)$, and so smooth in $\mathcal{R}\cup\hh$.

\section{Construction of the time integral $\phi$}
\label{sec:ConstructionOfTheTimeIntegralPhi}

On Schwarzschild backgrounds, Wald \cite{drimos} showed that 
given a solution to the wave equation with   compactly supported initial data on $t=0$, one can construct its time integral (see also the discussion at Section 3.2 in \cite{lecturesMD})

In this section we will reverse the logic. We will construct a (specific) solution $\psi$ to the wave equation by first constructing what will be its time integral $\phi$.  
We first prescribe initial data for $\phi$ on $\Sigma$. Let

\begin{equation}
\begin{split}
&\left.\phi\right|_{\Sigma_{0}\cup\Sigma_{1}}=1, \ \ \ \  \ \left.\phi\right|_{\Sigma_{3}}=0,\\
&\left.T\phi\right|_{\Sigma_{0}\cup\Sigma_{1}}=0, \ \ \ \left.T\phi\right|_{\Sigma_{3}}=0,\\
\end{split}\label{id1}
\end{equation}
and consider axisymmetric functions
\begin{equation*}
\begin{split}
\left.\phi\right|_{\Sigma_{2}}, \ \ \ \left.T\phi\right|_{\Sigma_{2}}\\
\end{split}
\end{equation*}
such that 
\begin{equation*}
\begin{split}
\left.\phi\right|_{\Sigma}, \ \ \ \left.T\phi\right|_{\Sigma}\\
\end{split}
\end{equation*}
are smooth and axisymmetric functions on $\Sigma$.
\begin{figure}[H]
	\centering
		\includegraphics[scale=0.17]{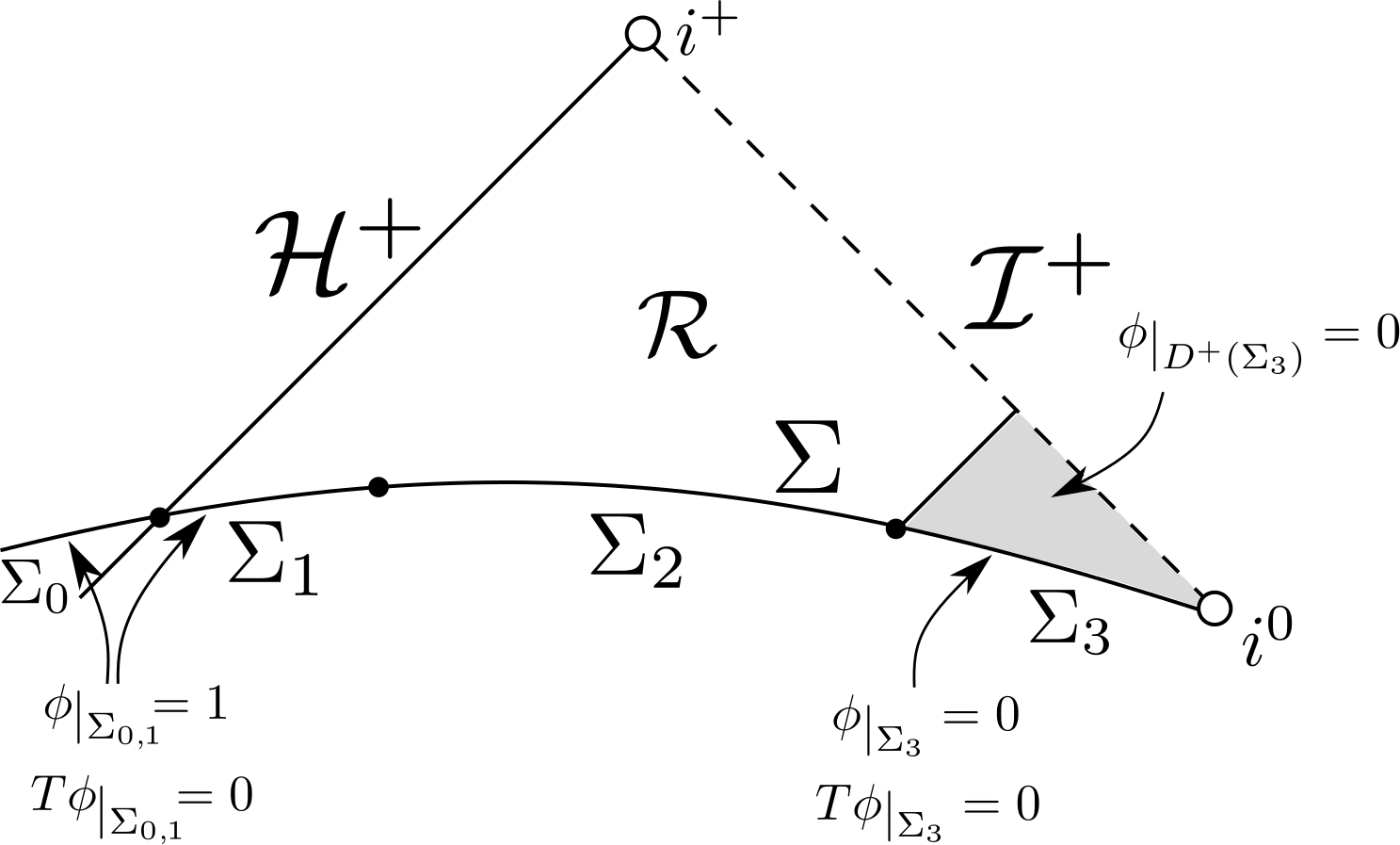}
	\label{fig:csind2}
\end{figure}
Then, there exists a unique solution $\phi$ to 
\begin{equation}
\Box_{g}\phi=0
\label{we1}
\end{equation}
which is smooth in $\r$. By virtue of the axisymmetry of the initial data of  $\phi$, the solution $\phi$ is axisymmetric.

\section{Construction of  $\psi$ with compactly supported initial data}
\label{sec:ConstructionOfPsiWithCompactlySupportedInitialData}

We define \[\boxed{\psi=T\phi\in C^{\infty}(D^{+}(\Sigma))\subset C^{\infty}(\r).}\] 
Since $T$ is Killing, we have $[\Box_{g},T]=0$ and hence 
\[\Box_{g}\psi=0\]in $D^{+}(\Sigma)$.
\begin{myp}
The data of $\psi$ on the hypersurface $\Sigma$ are compactly supported (and supported away from $\hh$). 
\label{p11}
\end{myp}
\begin{figure}[H]
	\centering
		\includegraphics[scale=0.17]{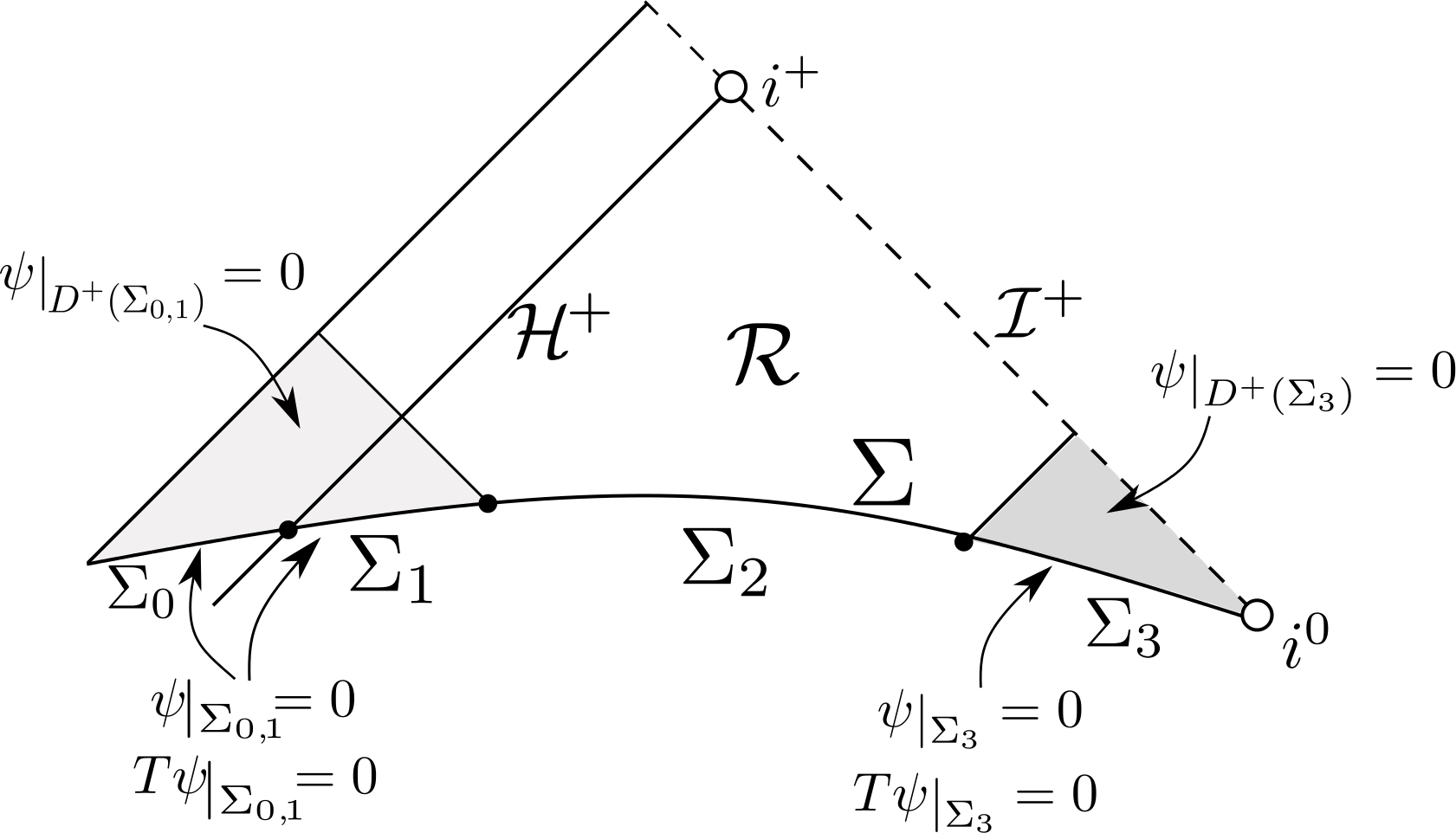}
	\label{fig:csind32}
\end{figure}
\begin{proof}
First note that since $[T,\Phi]=0$ the initial data of $\psi$ on $\Sigma$ are axisymmetric, and thus $\Phi\psi=0$ in $D^{+}(\Sigma)$. 
Clearly,
\[\left.\psi\right|_{\Sigma_{1}}=\left.T\phi\right|_{\Sigma_{1}}=0, \ \ \ \  \ \left.\psi\right|_{\Sigma_{3}}= \left.T\phi\right|_{\Sigma_{3}}=0. \]
We will show that we also have
\[\left.T\psi\right|_{\Sigma_{1}}=\left.TT\phi\right|_{\Sigma_{1}}=0, \ \ \ \  \ \left.T\psi\right|_{\Sigma_{3}}= \left.TT\phi\right|_{\Sigma_{3}}=0. \]
We start with $\Sigma_{1}$, where $\phi$ satisfies $\phi=1$, $T\phi=0$ and $\Box_{g}\phi=0$. The wave operator in $(v,r,\theta, \phi^{*})$ coordinates is (recall that $\partial_{v}=T$, $\partial_{r}=Y$ and $\partial_{\phi^{*}}=\Phi$):
\begin{equation*}
\begin{split}
\Box_{g}\phi=&\frac{a^{2}}{\rho^{2}}\sin^{2}\theta\left(TT\phi\right)+\frac{2(r^{2}+M^{2})}{\rho^{2}}\left(TY\phi\right)+\frac{\Delta}{\rho^{2}}(YY\phi)\\&+\frac{2a^{2}}{\rho^{2}}(T\Phi\phi)+\frac{2a}{\rho^{2}}(Y\Phi\phi)+\frac{2r}{\rho^{2}}(T\phi)+\frac{\Delta'}{\rho^{2}}(Y\phi) +\frac{1}{\rho^{2}}\lapp_{(\theta,\phi^{*})}\phi,
\end{split}
\end{equation*}
where $\lapp_{(\theta,\phi^{*})}\psi=\frac{1}{\sin\theta}\left(\partial_{\theta}\left[\sin\theta\cdot\partial_{\theta}\psi\right]\right)+\frac{1}{\sin^{2}\theta}\partial_{\phi^{*}}\partial_{\phi^{*}}\psi$ denotes the standard Laplacian on $\mathbb{S}^{2}$ with respect to $(\theta,\phi^{*})$ and $\rho^{2}=r^{2}+a^{2}\cos^{2}\theta$ and $\Delta=(r-M)^{2}$ (recall also that we consider $|a|=M$).

Furthermore, we consider the coordinate system $(r,\theta,\phi^{*})$ for $\Sigma$ (and, in fact, for $\Sigma_{\tau}$) and let $\partial_{\rho}$ be the coordinate vector field tangential to $\Sigma$ such that $\partial_{\rho}r=1$. Then,
\[\partial_{\rho}=h(r,\theta)T+Y,\]
where $h$ is a (strictly) positive function that depends only on $\Sigma$ given by
\[h=-\frac{g(Y,n_{\Sigma})}{g(T,n_{\Sigma})}.\]
\begin{figure}[H]
	\centering
		\includegraphics[scale=0.17]{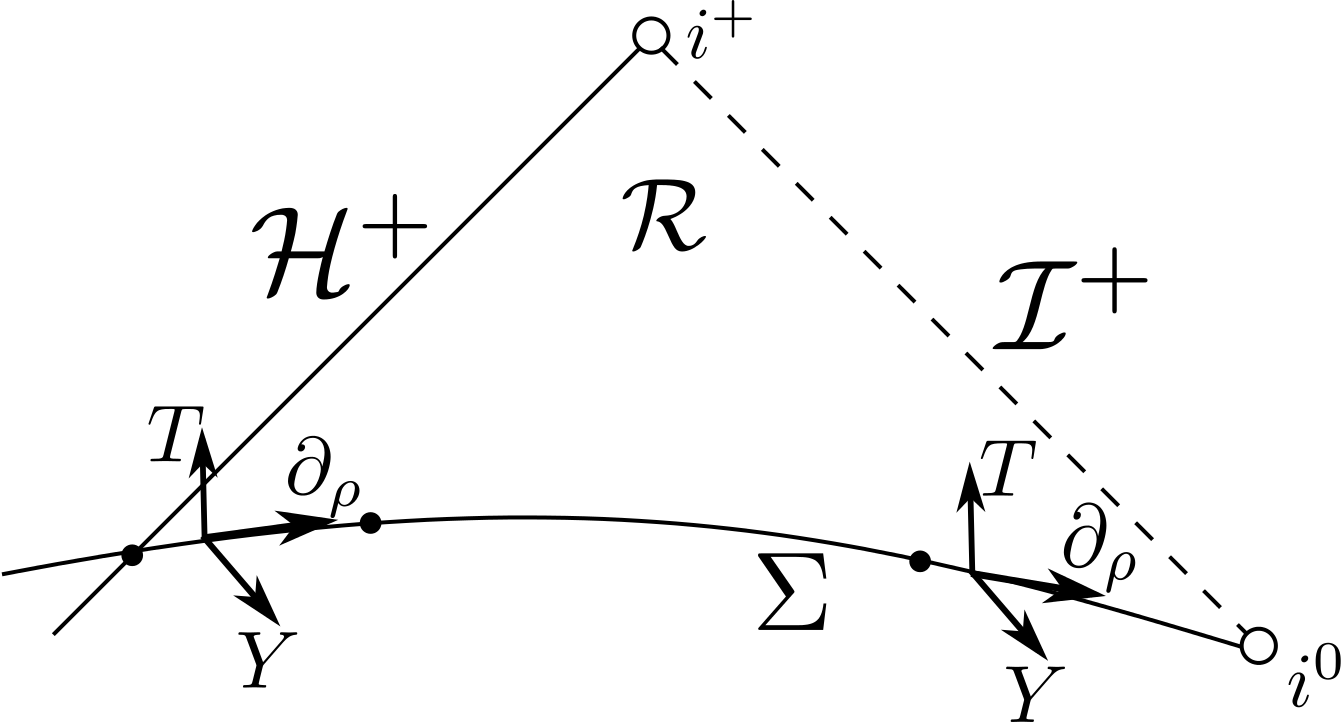}
	\label{fig:csind4}
\end{figure} Hence, the wave operator becomes
\begin{equation*}
\begin{split}
\rho^{2}\cdot (\Box_{g}\phi)=&M^{2}\sin^{2}\theta\left(TT\phi\right)+2(r^{2}+M^{2})\left(\partial_{\rho}T\phi\right)-2(r^{2}+M^{2})h\cdot\left(TT\phi\right)\\
&+\Delta\cdot\Big(\partial_{\rho}\partial_{\rho}\phi+h^{2}TT\phi-h'T\phi-2h\partial_{\rho}T\phi\Big)\\
&+2M^{2}(T\Phi\phi)+2M(\partial_{\rho}\Phi\phi)-2MhT\Phi\phi+2r(T\phi)+\Delta'(\partial_{\rho}\phi)-\Delta'hT\phi +\lapp\phi,
\end{split}
\end{equation*}
By restricting the above on $\Sigma_{0}\cup\Sigma_{1}$ and using \eqref{id1} (and that $\partial_{\rho}$ is tangential to $\Sigma$) we obtain
\[ \big(M^{2}\sin^{2}\theta-2(r^{2}+M^{2})h+\Delta h^{2}\big)\cdot(TT\phi)=0. \]
According to property \eqref{condition} of $\Sigma$, the coefficient of $TT\phi$  above is always negative. Therefore, $T\psi=TT\phi=0$ on $\Sigma_{0}\cup\Sigma_{1}$. The result on $\Sigma_{3}$ follows by the finite speed propagation property. 
\end{proof}

\bigskip

\section{Instabilities for $\psi$ along the horizon}
\label{sec:InstabilitiesForPsiAlongTheHorizon}

First note that since the initial data of $\psi$ are smooth, axisymmetric and compactly supported, in view of our previous stability results (see Theorem 5 of Section 3 in \cite{aretakis3}), we have that \[|\psi|\rightarrow 0,\ \ \ \ \ |T\psi|\rightarrow 0,\ \ \ \ \ |TT\psi|\rightarrow 0\]
along $\hh$ as $v\rightarrow +\infty$. Since the conserved quantities $H_{0}[\psi]=H_{0}[T\psi]=0$ we further obtain (see \cite{aretakis4})
\[\int_{S_{v}}Y\psi\rightarrow 0,\ \ \ \ \ \ \int_{S_{v}}YT\psi\rightarrow 0\]
 where $S_{v}$ are the $(\theta,\phi^{*})$ sections of $\hh$. Nonetheless,  we have the following instability results:
\begin{myp}
Let $\Sigma$ be a Cauchy hypersurface crossing $\hh$ as defined in Section \ref{sec:GeometricSetup} and $K=\Sigma\cap\left\{R_{1}\leq r \leq R_{2}\right\}$, where $M<R_{1}<R_{2}$. Then,  there exists a solution $\psi$ to the wave equation whose initial data are supported only on $K$ such that  
\[\sup_{S_{v}}|Y^{2}\psi|\nrightarrow 0, \ \ \ \ \ \sup_{S_{v}}|Y^{k}\psi|\rightarrow +\infty,\ k\geq 3, \]
along the horizon as the advanced time parameter $v$ tends to infinity.
\label{p21}
\end{myp}
\begin{proof}
We choose the constants $R_{1},R_{2}$ of Section \ref{sec:GeometricSetup} such that $\Sigma_{2}=K$ and consider the above construction for $\phi,\psi$. Recall that $\psi=T\phi$ in $\r$ and that the initial data for $\phi$ are trivial in $\Sigma_{3}$ (and so $\phi$ is trivial in $D^{+}(\Sigma_{3}$)) and smooth and axisymmetric everywhere on $\Sigma$ (up to and including $\hh$). Hence, we can apply our previous stability results to deduce that 
\[|\phi|\rightarrow 0,\ \ \ |T\phi|\rightarrow 0, \ \ \ |TT\phi|\rightarrow 0\]
along $\hh$ as $v\rightarrow +\infty$.
Recall that $\phi$ is axisymmetric and its conserved quantity is given by
\[H_{0}[\phi]=4M\int_{S_{v}}Y\phi+2\int_{S_{v}}\phi+M\int_{S_{v}}\sin^{2}\theta(T\phi)=2\neq 0,\]
by simply computing it at $\Sigma\cap\hh$. Hence, in view of the above stability results we obtain
\[\int_{S_{v}}Y\phi\rightarrow \frac{1}{2M}\neq 0.\]
  By restricting $Y\big(\rho^{2}\Box_{g}\phi\big)=0$ on the horizon $\hh$ we obtain
\begin{equation}
\begin{split}
4M^{2}\cdot YYT\phi=-M^{2}\sin^{2}\theta (YTT\phi)-6M(YT\phi)-2T\phi-2Y\phi-\lapp Y\phi.
\end{split}
\label{malista}
\end{equation}
Furthermore,  by restricting $\Box_{g}\psi=0$ on $\hh$ and multiplying with $\sin^{2}\theta$ we obtain \[-M^{2}\int_{S_{v}}\sin^{2}\theta (YT\psi)\rightarrow 0\] as $v\rightarrow+\infty$.
Therefore, 
\begin{equation*}
\begin{split}
\int_{S_{v}}4M^{2}\cdot YY\psi=&-M^{2}\int_{S_{v}}\sin^{2}\theta (YT\psi)-6M\int_{S_{v}}Y\psi-2\int_{S_{v}}T\phi-2\int_{S_{v}}Y\phi\rightarrow -\frac{1}{M}
\end{split}
\end{equation*}
as $v\rightarrow+\infty$, which shows that $\sup_{S_{v}}|Y^{2}\psi|\geq q>0$ asymptotically as $v\rightarrow +\infty$.

We next look at $Y^{3}\psi$. By restricting $Y^{2}(\rho^{2}\Box_{g}\psi)=0$ on $\hh$ we obtain that there exist smooth functions $\lambda_{i}$ such that
\begin{equation}
\begin{split}
TY^{3}\psi=&\lambda_{0}(M,\theta)T^{2}Y^{2}\psi+\lambda_{1}(M)TY\psi+\lambda_{2}(M)TY^{2}\psi+\lambda_{3}(M)T\psi+\lambda_{4}(M)YT\phi\\&+\lapp Y^{2}\psi-\lambda_{5}(M)Y^{2}\psi,   \end{split}
\label{int2}
\end{equation}
where $\lambda_{5}(M)>0$. Note that the integrals \[\int_{S_{v}}Y\psi,\ \int_{S_{v}}YY\psi, \ \int_{S_{v}}\psi, \ \int_{S_{v}}Y\phi, \ \int_{S_{v}}\lambda_{0}(M,\theta)TY^{2}\psi\] are uniformly bounded in $v$. The boundedness of the latter integral, in particular, follows by multiplying \eqref{malista} (where $\phi$ is replaced by $\psi$) with $\lambda_{0}(M,\theta)$ and using that the integrals of the form
\begin{equation}
\int_{S_{v}}\lambda(\theta)\cdot Y\psi
\label{int1}
\end{equation}
can be bounded as follows: By the wave equation $\Box_{g}\psi=0$ restricted on $\hh$ (and multiplied by $\lambda(\theta)$) we obtain:
\begin{equation}
T\Big(\lambda(\theta)T\psi+\lambda(\theta)Y\psi+\lambda(\theta)\psi\Big)+\lambda(\theta)\cdot\lapp\psi=0
\label{eq1}
\end{equation}
along the horizon. 
The boundedness of \eqref{int1} follows by integrating \eqref{eq1} along $\hh$ and using that 
\[\int_{v}\int_{S_{v}}\lambda(\theta)\cdot \lapp\psi=\int_{v}\int_{S_{v}}\lapp\lambda(\theta)\cdot \psi\] and that $\psi=T\phi$ and the previous stability results. 
Note also that the integral $\int_{S_{v}}\lambda(\theta)\cdot \lapp Y\psi$ reduces to the previous integrals by applying Green's identity on $S_{v}$.

By integrating \eqref{int2} along the horizon $\hh$ and using  that  $\int_{S_{v}}YY\psi\rightarrow -\frac{1}{4M^{3}}\neq 0$ we obtain 
\[\sup_{S_{v}}|Y^{3}\psi|\rightarrow +\infty\]
as $v\rightarrow +\infty$ along $\hh$. We can similarly show that 
\[\sup_{S_{v}}|Y^{k}\psi|\rightarrow +\infty\]
as $v\rightarrow +\infty$ along $\hh$, $k\geq 3$.

\end{proof}

\section{Remarks}
\label{sec:Remarks}

\subsection{Genericity of instabilities}
\label{sec:GenericityOfInstabilities}

By linearity one can immediately see that the instabilities can be inferred from generic initial data when genericity can be understood as follows.

Let $K=\Sigma\cap\left\{R_{1}\leq r\leq R_{2}\right\}$, where $M<R_{1}<R_{2}$. Let $\xi$ be a general smooth (not necessarily axisymmetric) solution of \eqref{we} which does not develop instabilities along the horizon $\hh$, i.e.
\[|\xi|\rightarrow 0,\ \ \ \ |T\xi|\rightarrow 0,\ \ \ \ |Y\xi|\rightarrow 0\]along $\hh$. Consider the solution $\psi$ constructed in Section \ref{sec:ConstructionOfPsiWithCompactlySupportedInitialData} such that $\Sigma_{2}=K$. For any constant $\epsilon>0$ the function $\xi+\epsilon\psi$ is a smooth solution  for which higher order transversal to $\hh$ derivatives blow up asympotically along $\hh$. In particular, if $\xi$ is not axisymmetric and is initially supported on $K$ then  $\xi+\epsilon\psi$ is also not axisymmetric  and is initially supported on $K$. This completes the argument\footnote{Note that the above argument effectively shows that the  codimension of the set of the ``stable  initial data'' is strictly positive.  }.

\subsection{Initial perturbations on $t=0$}
\label{sec:TheRelevanceOfTheHypersurfaceT0}

Consider the spacelike (complete) hypersurface $t=0$. Consider now $\Sigma$ to be a hypersurface which crosses $\hh$, is such that $\Sigma\cap \left\{r\geq R_{1}\right\}=\left\{t=0\right\}\cap\left\{r\geq R_{1}\right\}$ and satisfies the condition \eqref{condition}. We also consider the subsets $\Sigma_{0},\Sigma_{1},\Sigma_{2},\Sigma_{3}$ of $\Sigma$ as defined above (see also figure below). Let $\mathcal{R}$ be the domain of dependence of $\Sigma_{1}\cup\Sigma_{2}\cup\Sigma_{3}$.

Let also $\Sigma_{1}'=\left\{t=0\right\}\cap \left\{r\leq R_{1}\right\}$. Consider now compactly supported  initial data on $t=0$ such that they are trivial on $\Sigma_{1}'$ and $\Sigma_{3}$ and coincide with the data of Proposition \ref{p11} on $\Sigma_{2}$. Let $\psi_{1}$ be the unique solution that arises from such data.  By the domain of dependence property, the solution $\psi_{1}$    is trivial in the domain of dependence of $\Sigma_{1}'$. Hence the data of $\psi_{1}$ on $\Sigma_{1}$ is also trivial.

\begin{figure}[H]
	\centering
		\includegraphics[scale=0.11]{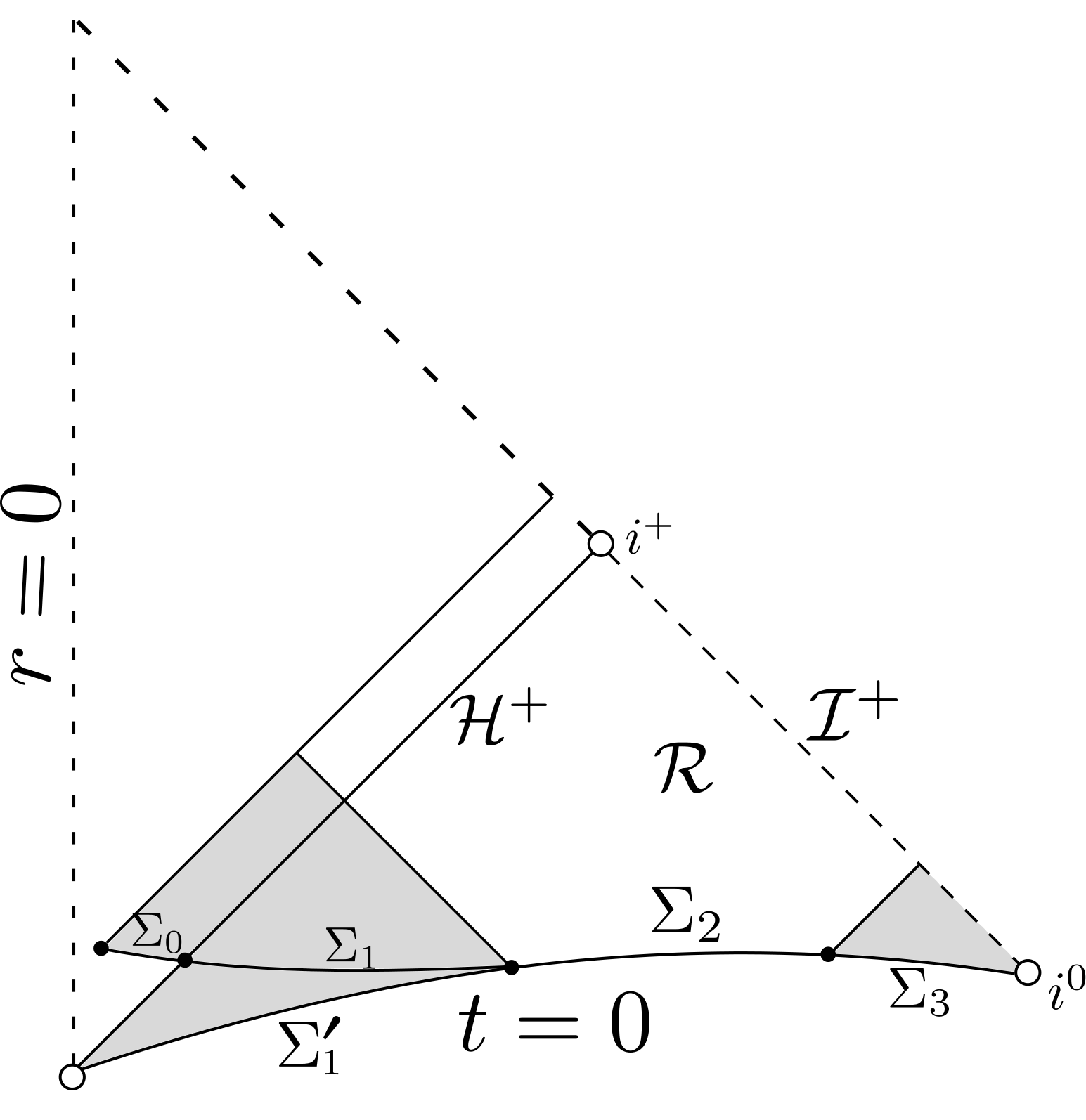}
	\label{fig:csind5}
\end{figure}

Therefore, the solutions $\psi_{1},\psi$ must necessarily coincide in the region $\mathcal{R}$.  By virtue of the smoothness of $\psi$ in $D^{+}(\Sigma)$ we have that  $\psi_{1}$ extend smoothly to  $\r\subset D^{+}(\Sigma)$ and the higher order transversal derivatives of $\psi_{1}$ blow up asymptotically along $\hh$.

Hence, by the above and the discussion in Remark \ref{sec:GenericityOfInstabilities} we deduce that for generic compactly supported initial data on $t=0$, the higher order transversal derivatives of $\psi_{1}$ blow up asymptotically along $\hh$.

\subsection{Extremal Reissner--Nordstr\"{o}m}
\label{sec:ExtremalReissnerNordstrOM}

Finally, let us mention that our method applies also for extremal Reissner--Nordstr\"{o}m black holes, i.e.~we can replace extremal Kerr with extremal Reissner--Nordstr\"{o}m in Theorem \ref{theo1}. In fact, in this case, one only requires $\Sigma_{0}$ to be spherically symmetric (without any additional restrictions such as \eqref{condition}). Note also that one needs to prescribe spherically symmetric initial data for $\phi$ on $\Sigma_{0}$.

\section{Acknowledgements}
\label{sec:Acknowledgements}

I would like to thank Mihalis Dafermos for his constant encouragement and help. I would also like to thank Sergio Dain and Gustavo Dotti for suggesting the problem and Harvey Reall for several helpful discussions. 

\bibliographystyle{acm}
\bibliography{../../../bibliography}

\begin{thebibliography}{10}

\bibitem{SA10}
{\sc Aretakis, S.}
\newblock The wave equation on extreme {R}eissner--{N}ordstr\"om black hole
  spacetimes: stability and instability results.
\newblock {\em ar{X}iv:1006.0283\/} (2010).

\bibitem{aretakis1}
{\sc Aretakis, S.}
\newblock Stability and instability of extreme {R}eissner--{N}ordstr\"om black
  hole spacetimes for linear scalar perturbations {I}.
\newblock {\em Commun. Math. Phys. 307\/} (2011), 17--63.

\bibitem{aretakis2}
{\sc Aretakis, S.}
\newblock Stability and instability of extreme {R}eissner--{N}ordstr\"om black
  hole spacetimes for linear scalar perturbations {II}.
\newblock {\em Ann. Henri Poincar\'{e} 12\/} (2011), 1491--1538.

\bibitem{aretakis3}
{\sc Aretakis, S.}
\newblock Decay of axisymmetric solutions of the wave equation on extreme
  {K}err backgrounds.
\newblock {\em J. Funct. Analysis 263\/} (2012), 2770--2831.

\bibitem{aretakis4}
{\sc Aretakis, S.}
\newblock Horizon instability of extremal black holes.
\newblock {\em ar{X}iv:1206.6598\/} (2012).

\bibitem{other2}
{\sc Bi\v{c}\'{a}k, J.}
\newblock Gravitational collapse with charge and small asymmetries {I}:
  {S}calar perturbations.
\newblock {\em Gen. Rel. Grav. 3\/} (1972), 331--349.

\bibitem{bizon2012}
{\sc Bizon, P., and Friedrich, H.}
\newblock A remark about the wave equations on the extreme
  {R}eissner--{N}ordstr\"om black hole exterior.
\newblock {\em Class. Quantum Grav. 30\/} (2013), 065001.

\bibitem{blu1}
{\sc Blue, P., and Soffer, A.}
\newblock Phase space analysis on some black hole manifolds.
\newblock {\em Journal of Functional Analysis 256\/} (2009), 1--90.

\bibitem{lecturesMD}
{\sc Dafermos, M., and Rodnianski, I.}
\newblock Lectures on black holes and linear waves.
\newblock {\em ar{X}iv:0811.0354\/}.

\bibitem{redshift}
{\sc Dafermos, M., and Rodnianski, I.}
\newblock The redshift effect and radiation decay on black hole spacetimes.
\newblock {\em Comm. Pure Appl. Math. 62\/} (2009), 859--919, ar{X}iv:0512.119.

\bibitem{tria}
{\sc Dafermos, M., and Rodnianski, I.}
\newblock The black hole stability problem for linear scalar perturbations.
\newblock {\em Proceedings of the 12 Marcel Grossmann Meeting, edited by T.
  Damour et al (ed.), World Scientific, Singapore\/} (2011), 132--189,
  ar{X}iv:1010.5137.

\bibitem{dd2012}
{\sc Dain, S., and Dotti, G.}
\newblock The wave equation on the extreme {R}eissner--{N}ordstr\"om black
  hole.
\newblock {\em ar{X}iv:1209.0213\/} (2012).

\bibitem{jacobson}
{\sc Jacobson, T.}
\newblock Semiclassical decay of near-extremal black holes.
\newblock {\em Phys.Rev. D 57\/} (1998), 4890--4898.

\bibitem{wa1}
{\sc Kay, B., and Wald, R.}
\newblock Linear stability of {S}chwarzschild under perturbations which are
  nonvanishing on the bifurcation 2-sphere.
\newblock {\em Class. Quantum Grav. 4\/} (1987), 893--898.

\bibitem{hm2012}
{\sc Lucietti, J., Murata, K., Reall, H.~S., and Tanahashi, N.}
\newblock On the horizon instability of an extreme {R}eissner--{N}ordstr\"om
  black hole.
\newblock {\em arXiv:1212.2557\/} (2012).

\bibitem{hj2012}
{\sc Lucietti, J., and Reall, H.}
\newblock Gravitational instability of an extreme kerr black hole.
\newblock {\em Phys. Rev. D86:104030\/} (2012).

\bibitem{marolf}
{\sc Marolf, D.}
\newblock The danger of extremes.
\newblock {\em Gen. Rel. Grav. 42\/} (2010), 2337--2343.

\bibitem{murata2012}
{\sc Murata, K.}
\newblock Instability of higher dimensional extreme black holes.
\newblock {\em Class. Quantum Grav. 30\/} (2013), 075002.

\bibitem{regge}
{\sc Regge, T., and Wheeler, J.}
\newblock Stability of a {S}chwarzschild singularity.
\newblock {\em Phys. Rev. 108\/} (1957), 1063--1069.

\bibitem{drimos}
{\sc Wald, R.~M.}
\newblock Note on the stability of the {S}chwarzschild metric.
\newblock {\em J. Math. Phys. 20\/} (1979), 1056--1058.

\end{thebibliography}

\end{document}